\documentclass[reqno]{amsart}
\usepackage[T1]{fontenc}
\usepackage{epsfig,psfrag}
\usepackage{version,bbm}
\usepackage{graphics,color}

\def\R{\mathbb{R}}
\def\C{\mathbb{C}}
\def\E{\mathbb{E}}
\def\Z{\mathbb{Z}}
\def\ind{\mathbbm{1}}

\def\atan{\mathrm{atan}}
\def\var{\mathrm{var}}
\graphicspath{{Fig/}}
\newtheorem{lemma}{Lemma}
\newtheorem{theorem}{Theorem}
\newtheorem{proposition}{Proposition}
\newtheorem{corollary}{Corollary}
\newtheorem{condition}{Condition}

\title[Analysis of  a non work-conserving GPS queue]{Analysis of a  non-work conserving  Generalized Processor Sharing queue}
\author[F. Guillemin]{Fabrice Guillemin}
\address{Fabrice Guillemin, Orange Labs, 2 Avenue Pierre Marzin, 22300 Lannion}
\email{Fabrice.Guillemin@orange.com}

\date{}

\begin{document}

\begin{abstract}
We consider  in this paper a non work-conserving Generalized Processor Sharing (GPS) system composed of two queues with Poisson arrivals and exponential service times.  Using general results due to Fayolle \emph{et al}, we first establish the stability condition for this system. We then determine the functional equation satisfied by the generating function of the numbers of jobs  in both queues and the associated Riemann-Hilbert problem. We prove the existence and the uniqueness of the solution. This allows us to completely characterize the system, in particular to compute the empty queue probability.  We finally derive the tail asymptotics of the number of jobs in one queue.
\end{abstract}

\maketitle

\section{Introduction}

The Generalized Processor Sharing (GPS) queue is a well known queuing system which has  extensively been studied in the literature in the past decades  for two queues in parallel, see for instance  \cite{Fayolle,GP}. This model has recently  gained renewed interest  in connection with bandwidth sharing in telecommunication networks for elastic flows \cite{Jim9}. A GPS system consists of queues in parallel served according to some weights by a  server with capacity $r$, which usually assumed to be work-conserving. To be more specific, if there are $M$ queues, the service rate at time $t$ of the $i$th queue when not empty is $\phi_i r/\sum_{j=1}^N \phi_j\ind_{\{N_j(t)>0\}}$, where $N_j(t)$ is the number of jobs  in queue $\#j$ at time $t$ and the weights $\phi_j\in (0,1)$ are such that $\sum_{j=1}^M\phi_j=1$. The system is work-conserving in the sense that the global service rate for all queues, which are not empty, is equal to the server rate  $r$.

In this paper, we consider a GPS system composed of two queues in parallel but we assume that $\phi_1+\phi_2>1$. In this case,  when one queue is empty, the service rate for the other queue if not empty is  $r/(\phi_1+\phi_2)$, which less than $r$, and when both queues are not empty, the service rate of queue $\#j$ is $\phi_j r /(\phi_1+\phi_2)$. The system is thus \emph{non work-conserving}. For this system, even the  empty queue probability is  not known since the classical ``$(1-\rho)$'' formula does not hold. In the following, we assume that jobs arrive at queue $\#i$ according to a Poisson process with rate $\lambda_i>0$ and require exponentially distributed service times with mean $1/\nu_i$. The objective of the paper is to explicitly compute under stability conditions the generating function of the numbers of jobs in the stationary regime.

The process $(N_1(t),N_2(t))$ describing the numbers of jobs in queues $\#1$ and $\#2$ is a random walk in the quarter plane. This kind of process has  extensively been studied in the technical   literature, notably see the book by Fayolle \emph{et al} \cite{FIM}, which heavily relies on the analysis of  Riemann surfaces associated with the quadratic kernel appearing when computing  generating functions. Recently, this kind of system has gained renewed interest, notably for the study of asymptotics  \cite{JvL,JvL2,Li11,KM11,Miyazawa11}; see also \cite{KR} for the analysis of associated Green functions.

Under stability conditions, the generating function of the stationary numbers of jobs in queues $\#1$ and $\#2$, defined by $P(x,y)=\E(x^{N_1}y^{N_2})$ for for complex $x$ and $y$ such that $|x|\leq 1$ and $|y|\leq 1$, satisfies a functional equation of the form 
\begin{equation}
\label{eqfund}
h_1(x,y) P(x,y) = h_2(x,y)P(x,0) + h_3(x,y)P(0,y)+h_4(x,y)P(0,0),
\end{equation}
where $h_i(x,y)$, $i=1,...,4$ are at most quadratic polynomials in variables $x$ and $y$. To determine the functions $P(x,0)$ and $P(0,y)$ we show that these functions can be analytically continued to some disks and satisfy Riemann-Hilbert problems of the following form: Find a function $f(z)$ analytic in a disk $D$ and satisfies 
\begin{equation}
\label{RHPgen1}
\Re\left(a(z)f(z)\right)=c(z)
\end{equation}
for $z$ on the circle $C$ delineating the disk $D$, where $a(z)$ and $c(z)$ are some functions depending on the data of the system.

Instead of using the method presented in \cite{FIM} involving Riemann surfaces and a Carleman problem (see Theorem~5.4.3 in that book) or the conformal mapping method successfully used by Blanc in \cite{Blanc}, we show that the above problem can be reduced to a Riemann-Hilbert problem of the following form: Find a function $\phi(z)$ which is sectionally analytic with respect to a circle $C$ and satisfies on the boundary condition
\begin{equation}
\label{RHPgen2}
\phi^i(z) = \alpha(z)\phi^e(z) +\beta(z)
\end{equation}
where $\phi^i(z)$ (resp. $\phi^e(z)$) is the interior (resp. exterior) limit of the function $\phi(z)$ at the circle $C$. We shall show in the following that those latter Riemann-Hilbert problems are with index 0 for the model under consideration. The general form of the solution to such a Riemann-Hilbert problem is given in \cite{Lions}.

As mentioned in \cite{Musk} when $C$ is the unit circle, a solution to Problem~\eqref{RHPgen2} may not be a solution to Problem~\eqref{RHPgen1} and the solution to the former may be written as
$$
f(z) = \frac{1}{2} \left( \phi(z) + \phi^*(z) \right)
$$
where 
$$
\phi^*(z) = \overline{\phi\left(\frac{1}{\overline{z}}\right)},
$$
where  $\phi (z) $ is the solution to Problem~\eqref{RHPgen2}. However, by using the expansion in power series of $z$ at the origin and the fact that radius of convergence of this series is larger than the radius of circle $C$, we are  able to show that the solution to Problem~\eqref{RHPgen2} is indeed the solution to Problem~\eqref{RHPgen1}.

The organization of this paper is as follows: In Section~\ref{model}, we describe the model and establish some preliminary results,  in particular the coefficients of the functional equation~\eqref{eqfund} together with the stability with the stability condition for the system. In Section~\ref{stability}, we establish the Riemann-Hilbert problems associated with the system.  These problems are solved in Section~\ref{rhp} so that we obtain an explicit expression for the generating function $P(x,y)$. This explicit form  is finally  used in Section~\ref{asymptotics} to derive  queue asymptotics. Some conluding remarks are presented in Section~\ref{conclusion}.

\section{Model description and preliminary results}
\label{model}

\subsection{Stability condition}

By setting $\mu_i=\nu_ir/(\phi_1+\phi_2)$, it is easily checked that the process $(N_1(t),N_2(t))$ describing the numbers of jobs  in queues $\#1$ and $\#2$ is a random walk in the quarter plane  with non null transition rates
\begin{align*}
&r_{1,0} = \lambda_1, \; r_{0,1} = \lambda_2, \; r_{-1,0} = \phi_1\mu_1, \; r_{0,-1} = \phi_2 \mu_2,\\
&r'_{1,0} = \lambda_1, \; r'_{0,1} = \lambda_2, \; r'_{-1,0} = \mu_1,\\
&r''_{1,0} = \lambda_1, \; r''_{0,1} = \lambda_2, \; r''_{0,-1} =  \mu_2, \\
&r'''_{1,0} = \lambda_1, \; r'''_{0,1} = \lambda_2,
\end{align*}
where $r_{k,l}$ are the transition rate from state $(i,j)$ to state $(i+k,j+l)$ for $i,j>0$, $r'_{k,l}$ is the transition rate from state $(i,0)$ to state $(i+k,l)$,  $r''_{k,l}$ is the transition rate from state $(0,j)$ to state $(k,j+l)$, and $r'''_{k,l}$ is the transition rate from $(0,0)$ to $(k,l)$.

By using the results in \cite{FayolleStab} (or alternatively the corrected version in \cite{KM12}), the system is stable if and only if
\begin{equation}
\label{stabcondex}
\rho_1+\frac{1-\phi_1}{\phi_2}\rho_2<1 \quad \mbox{and} \quad \frac{1-\phi_2}{\phi_1}\rho_1+\rho_2<1,
\end{equation}
where $\rho_1=\lambda_1/\mu_1$ and $\rho_2=\lambda_2/\mu_2$. In terms of the initial data of the system, this condition reads
$$
\frac{\lambda_1}{\nu_1}+\frac{1-\phi_1}{\phi_2}\frac{\lambda_2}{\nu_2}<\frac{r}{\phi_1+\phi_2}  \quad \mbox{and} \quad \frac{1-\phi_2}{\phi_1}\frac{\lambda_1}{\nu_1}+\frac{\lambda_2}{\nu_2}<\frac{r}{\phi_1+\phi_2}.
$$
Note that when $\phi_1+\phi_2=1$, the stability condition reads
\begin{equation}
\label{stabGPS}
\rho_1+\rho_2<1,
\end{equation}
which is the usual condition for the work conserving GPS system.


\subsection{Fundamental equation}

We assume that Condition~\eqref{stabcondex} holds so that the system is stable. Let $p(n_1,n_2)$ be the stationary probability that there are $n_1$ jobs in queue \#1 and $n_2$ jobs in queue \#2. The balance equations read for $n_1 \geq 0$ and $n_2 \geq 0$:
\begin{multline}
\label{be}
 \lambda_1 \ind_{\{ n_1>0\}} p(n_1-1,n_2) +\lambda_2\ind_{\{ n_2>0\}}p(n_1,n_2-1)  +\phi_1 \mu_1 \ind_{\{ n_2=0\}} p(n_1+1,n_2) \\
+  \phi\mu_2\ind_{\{ n_1>0\}}p(n_1,n_2+1)+ \mu_1 \ind_{\{ n_2=0\}} p(n_1+1,n_2) + \mu_2 \ind_{\{ n_1=0\}} p(n_1,n_2+1) =\\ 
(\lambda_1+\lambda_2+\mu_1 \ind_{\{n_1>0,n_2=0\}} + \mu_2   \ind_{\{n_1=0,n_2>0\}}  +(\phi_1\mu_1+\phi_2\mu_2)\ind_{\{n_1>0,n_2>0\}})  p(n_1,n_2)
\end{multline}

Multiplying the balance equations~\eqref{be} by $x^{n_1}y^{n_2}$ and summing them up, we obtain the functional equation~\eqref{eqfund} for the generating function $P(x,y)$ with
\begin{eqnarray*}
h_1(x,y) &=& -\lambda_1x^2y-\lambda_2xy^2+(\lambda_1+\lambda_2+\phi_1\mu_1+\phi_2\mu_2)xy-\phi_1\mu_1y-\phi_2\mu_2x,\\
h_2(x,y) &=& \phi_2\mu_2x(y-1) - (1-\phi_1)\mu_1y(x-1),\\
h_3(x,y) &=& \phi_1\mu_1y(x-1)-(1-\phi_2)\mu_2x(y-1),\\
h_4(x,y) &=& (1-\phi_1)\mu_1y(x-1)+(1-\phi_2)\mu_2x(y-1).
\end{eqnarray*}
 It is worth noting that 
 \begin{equation}
\label{h2h3h4}
(1-\phi_2)h_2(x,y) +(1-\phi_1)h_3(x,y) +(1-\phi_1-\phi_2)h_4(x,y)=0.
\end{equation}

\subsection{Zero pairs of the kernel}

For fixed $y$, the equation $h_1(x,y) = 0$ in variable $x$ has two roots
\begin{equation}
\label{defXpm}
X_\pm(y) = \frac{-(\lambda_2y^2-(\lambda_1+\lambda_2+\phi_1\mu_1+\phi_2\mu_2)y+\phi_2\mu_2)\pm \sqrt{\Delta_2(y)}}{2\lambda_1y},
\end{equation}
where
$$
\Delta_2(y) = (\lambda_2y^2-(\lambda_1+\lambda_2+\phi_1\mu_1+\phi_2\mu_2)y+\phi_2\mu_2)^2-4\lambda_1\phi_1\mu_1 y^2.
$$

It is easily checked that the discriminant $\Delta_2(y)$ has four real roots $0<y_1<y_2\leq 1 \leq y_3<y_4$ given by
\begin{eqnarray*}
y_1 &=& \frac{(\lambda_1+\lambda_2+\phi_1\mu_1+\phi_2\mu_2+2\sqrt{\lambda_1\phi_1\mu_1})-\sqrt{ \delta_X^{(1)} }}{2\lambda_2},  \\
y_2 &=&  \frac{(\lambda_1+\lambda_2+\phi_1\mu_1+\phi_2\mu_2-2\sqrt{\lambda_1\phi_1\mu_1})-\sqrt{  \delta_X^{(2)}}}{2\lambda_2},\\
y_3 &=&  \frac{(\lambda_1+\lambda_2+\phi_1\mu_1+\phi_2\mu_2-2\sqrt{\lambda_1\phi_1\mu_1})+\sqrt{  \delta_X^{(2)}}}{2\lambda_2},\\
y_4 &=& \frac{(\lambda_1+\lambda_2+\phi_1\mu_1+\phi_2\mu_2+2\sqrt{\lambda_1\phi_1\mu_1})+\sqrt{   \delta_X^{(1)} }}{2\lambda_2},
\end{eqnarray*}
where
\begin{eqnarray*}
\delta_X^{(1)} &=& (\lambda_1+\lambda_2+\phi_1\mu_1+\phi_2\mu_2+2\sqrt{\lambda_1\phi_1\mu_1})^2-4\lambda_2\phi_2\mu_2 ,\\
\delta_X^{(2)} &=& (\lambda_1+\lambda_2+\phi_1\mu_1+\phi_2\mu_2-2\sqrt{\lambda_1\phi_1\mu_1})^2-4\lambda_2\phi_2\mu_2.
\end{eqnarray*}
Note that $y_2=1$ when $\phi_1=\rho_1$ and $\rho_2\leq \phi_2$, and $y_3=1$ when $\phi_1=\rho_1$ and $\rho_2\geq \phi_2$. In addition, the derivative polynomial $\Delta_1'(y)$ has three real roots $y_1^*$, $y_2^*$ and $y_3^*$ such that $0<y_1<y_1^*< y_2<y_2^*< y_3<y^*_3<y_4$. 

The functions $X_\pm(y)$ defined by Equation~\eqref{defXpm} are well defined for $y\in \R\setminus([y_1,y_2]\cup[y_3,y_4])$. By considering the analytic continuation of the square root of the polynomial $\Delta_2(y)$ in $\C\setminus ([y_1,y_2]\cup[y_3,y_4])$, denoted by $\varrho_2(y)$ and such that $\varrho_2(0) \equiv \sqrt{\Delta_2(0)}>0$ (see \cite{FIM,Nova} for details), we can define the functions
\begin{eqnarray}
X^*(y) &=& \frac{-(\lambda_2y^2-(\lambda_1+\lambda_2+\phi_1\mu_1+\phi_2\mu_2)y+\phi_2\mu_2)+\varrho_2(y)}{2\lambda_1y},\label{defXstar}\\
X_*(y) &=& \frac{-(\lambda_2y^2-(\lambda_1+\lambda_2+\phi_1\mu_1+\phi_2\mu_2)y+\phi_2\mu_2)-\varrho_2(y)}{2\lambda_1y}. \label{defXstarstar}
\end{eqnarray}
The function $X^*(y)$ is analytic in $\C\setminus([y_1,y_2]\cup[y_3,y_4])$ and the function $X_*(y)$  is meromorphic  in $\C\setminus([y_1,y_2]\cup[y_3,y_4])$ with a single pole at 0. The function $X^*(y)$ analytically continues in the whole of $\C\setminus([y_1,y_2]\cup[y_3,y_4])$ the function $X_+(y)$ defined for $y<y_1$. Similarly, the function $X_*(y)$ meromorphically continues in the same domain the function $X_-(y)$ defined for $y<y_1$.

Similarly, for fixed $x$, the zeros of the kernel $h_1(x,y)$ in variable $y$ are given by
$$
Y_\pm(x) = \frac{-(\lambda_1x^2-(\lambda_1+\lambda_2+\phi_1\mu_1+\phi_2\mu_2)x+\phi_1\mu_1)\pm \sqrt{\Delta_1(x)}}{2\lambda_2x}, 
$$
where 
$$
\Delta_1(x) = (\lambda_1x^2-(\lambda_1+\lambda_2+\phi_1\mu_1+\phi_2\mu_2)x+\phi_1\mu_1)^2-4\lambda_2\mu_2\phi_2 x^2.
$$
The discriminant $\Delta_1(x)$ has four real roots $0<x_1<x_2\leq 1<x_3 < x_4$ given by
\begin{eqnarray*}
x_1 &=& \frac{(\lambda_1+\lambda_2+\phi_1\mu_1+\phi_2\mu_2+2\sqrt{\lambda_2\phi_2\mu_2})-\sqrt{ \delta_Y^{(1)} }}{2\lambda_1},  \\
x_2 &=&  \frac{(\lambda_1+\lambda_2+\phi_1\mu_1+\phi_2\mu_2-2\sqrt{\lambda_2\phi_2\mu_2})-\sqrt{  \delta_Y^{(2)}}}{2\lambda_1},\\
x_3 &=&  \frac{(\lambda_1+\lambda_2+\phi_1\mu_1+\phi_2\mu_2-2\sqrt{\lambda_2\phi_2\mu_2})+\sqrt{  \delta_Y^{(2)}}}{2\lambda_1},\\
x_4 &=& \frac{(\lambda_1+\lambda_2+\phi_1\mu_1+\phi_2\mu_2+2\sqrt{\lambda_2\phi_2\mu_2})+\sqrt{   \delta_Y^{(1)} }}{2\lambda_1},
\end{eqnarray*}
where
\begin{eqnarray*}
\delta_Y^{(1)} &=& (\lambda_1+\lambda_2+\phi_1\mu_1+\phi_2\mu_2+2\sqrt{\lambda_2\phi_2\mu_2})^2-4\lambda_1\phi_1\mu_1, \\
\delta_Y^{(2)} &=& (\lambda_1+\lambda_2+\phi_1\mu_1+\phi_2\mu_2-2\sqrt{\lambda_2\phi_2\mu_2})^2-4\lambda_1\phi_1\mu_1.
\end{eqnarray*}
Note that $x_2=1$ when $\phi_2=\rho_2$ and $\rho_1\leq \phi_1$, and $x_3=1$ when $\phi_2=\rho_2$ and $\rho_1\geq \phi_1$. The derivative polynomial $\Delta'_1(x)$ has three real roots denoted by $x_1^*$, $x_2^*$ and $x_3^*$ such that $x_1<x_1^* < x_2< x_2^* < x_3 < x_3^*<x_4$.

The functions $Y_\pm(x)$ are defined in $\R\setminus([x_1,x_2]\cup[x_3,x_4])$. By considering the analytic continuation of the square root of the polynomial $\Delta_1(x)$ in $\C\setminus ([x_1,x_2]\cup[x_3,x_4])$, denoted by $\varrho_1(x)$ and such that $\varrho_1(0) \equiv \sqrt{\Delta_1(0)}>0$, we can define the functions
\begin{eqnarray}
\label{defYstar}
Y^*(x) &=&  \frac{-(\lambda_1x^2-(\lambda_1+\lambda_2+\phi_1\mu_1+\phi_2\mu_2)y+\phi_1\mu_1)+\varrho_1(x)}{2\lambda_2x}, \\
Y_*(x) &=&  \frac{-(\lambda_1x^2-(\lambda_1+\lambda_2+\phi_1\mu_1+\phi_2\mu_2)y+\phi_1\mu_1)-\varrho_1(x)}{2\lambda_2x}.\label{defYstarstar}
\end{eqnarray}
The function $Y^*(x)$ is analytic in $\C\setminus([x_1,x_2]\cup[x_3,x_4])$ and the function $Y_*(x)$  is meromorphic  in $\C\setminus([x_1,x_2]\cup[x_3,x_4])$ with a single pole at 0. The function $Y^*(x)$ analytically continues in the whole of $\C\setminus([x_1,x_2]\cup[x_3,x_4])$ the function $Y_+(x)$ defined for $x<x_1$. Similarly, the function $Y_*(x)$ meromorphically continues in the same domain the function $Y_-(x)$ defined for $x<x_1$.

Let $D\left(0,\sqrt{\frac{\phi_2}{\rho_2}}\right)$ denote the disk with center 0 and radius $\sqrt{\frac{\phi_2}{\rho_2}}$ and  $D\left(0,\sqrt{\frac{\phi_1}{\rho_1}}\right)$ be the disk with center 0 and radius $\sqrt{\frac{\phi_1}{\rho_1}}$. A classical result in the theory of random walks in the quarter plane is the following conformal mapping result; see  \cite{FIM} for details.

\begin{proposition}
\label{conformal}
The function $X^*(y)$ is a conformal mapping from the open set  $D\left(0,\sqrt{\frac{\phi_2}{\rho_2}}\right)\setminus[y_1,y_2]$ onto  the open set  $D\left(0,\sqrt{\frac{\phi_1}{\rho_1}}\right)\setminus [x_1,x_2]$. The reciprocal function is $Y^*(x)$.
\end{proposition}

To conclude this section, let us note  that $X_-(1)\leq X_+(1)$ and we have $X_-(1)<1$ if and only if $\rho_1>\phi_1$ and then $X_-(1) = \phi_1/\rho_1$. Similarly, $Y_-(1)<1$ if and only if $\rho_2>\phi_2$ and then $Y_-(1) =\phi_2/\rho_2$. Moreover, we have from Equation~\eqref{eqfund}
\begin{eqnarray}
P(x,1) &=& \frac{-(1-\phi_1)P(x,0)+\phi_1P(0,1)+(1-\phi_1)P(0,0)}{\phi_1-\rho_1x}, \label{Px1}\\
P(1,y) &=& \frac{\phi_2P(1,0)-(1-\phi_2)P(0,y)+(1-\phi_2)P(0,0)}{\phi_2-\rho_2y}.\label{P1y}
\end{eqnarray}
The normalizing condition $P(1,1)=1$ implies that 
\begin{eqnarray*}
-(1-\phi_1)P(1,0)+\phi_1P(0,1)+(1-\phi_1)P(0,0)&=&\phi_1-\rho_1 ,\\
 \phi_2P(1,0)-(1-\phi_2)P(0,1)+(1-\phi_2)P(0,0) &=& \phi_2-\rho_2,
\end{eqnarray*}
from which we deduce that
\begin{eqnarray}
P(1,0) &=& \frac{-\phi_1+(1-\phi_2)\rho_1+\phi_1\rho_2+(1-\phi_2)P(0,0)}{1-\phi_1-\phi_2},\label{P10}\\
P(0,1) &=& \frac{-\phi_2+(1-\phi_1)\rho_2+\phi_2\rho_1+(1-\phi_1)P(0,0)}{1-\phi_1-\phi_2}.\label{P01}
\end{eqnarray}

\section{Associated Riemann-Hilbert problem}
\label{stability}
\subsection{Analytic continuation}

Assume that the system is stable. The function $X^*(y)$ defined by Equation~\eqref{defXstar} is such that $X^*(0)=0$. Hence, there exists a neighborhood $V_Y(0)$ of 0 such that $X^*(V_Y(0))\subset D(0,1)$ and hence $P(X^*(y),0)$ is defined since $P(x,0)$ shall be defined in the closed disk $\overline{D(0,1)}$. From Equation~\eqref{eqfund}, we have for $y \in V_Y(0)$
\begin{equation}
\label{extPx0}
P(X^*(y),0) = -\frac{h_3(X^*(y),y)P(0,y)+h_4(X^*(y),y)P(0,0)}{h_2(X^*(y),y)}.
\end{equation}

Similarly, we can use the function $Y^*(x)$ defined by Equation~\eqref{defYstar} and find a neighborhood $V_X(0)$ of 0 such that we get from Equation~\eqref{eqfund} for $x \in V_X(0)$
\begin{equation}
\label{extP0y}
P(0,Y^*(x)) = -\frac{h_2(x,Y^*(x))P(x,0)+h_4(x,Y^*(x))P(0,0)}{h_3(x,Y^*(x))}.
\end{equation}

The above equations can be used to prove the following result.

\begin{proposition}
\label{analyticext}
If the system is stable, the functions $P(x,0)$ and $P(0,y)$ can analytically be continued in the disks $D\left(0,\sqrt{\frac{\phi_1}{\rho_1}}\right)$ and $D\left(0,\sqrt{\frac{\phi_2}{\rho_2}}\right)$, respectively. Both functions are in addition continuous in the respective closed disk.
\end{proposition}

\begin{proof}
If $\phi_1\leq \rho_1$ and  the system is stable, then  the function $P(x,0)$ is analytic in  the disk $D\left(0,\sqrt{\frac{\phi_1}{\rho_1}}\right) \subset D(0,1)$ and  continuous in the closed disk $\overline{D\left(0,\sqrt{\frac{\phi_1}{\rho_1}}\right)}$. 

If $\phi_1>\rho_1$, the function $X^*(y)$ is decreasing from $\sqrt{\frac{\phi_1}{\rho_1}}$ to 1 when the $y$ traverses the segment $[y_2,1]$ and when the system is stable, the equation $h_2(X^*(y),y)=0$ has no solution in $[y_2,1]$. As a matter of fact, we can prove that under Condition~\eqref{stabcondex}, the only solutions to $h_1(x,y)=h_2(x,y)=0$ in $\overline{D(0,1)}\times\overline{D(0,1)}$ are $(0,0)$ and $(1,1)$. Indeed, $h_1(x,y)=0$ reads 
$$
(\phi_1\mu_1-\lambda_1 x)y(x-1) = (\lambda_2y-\phi_2\mu_2)x(y-1)
$$
and by combining with $h_2(x,y)=0$, we obtain if $(x,y) \notin \{ (0,0),(1,1)\}$ that $(x,y)$ has to satisfy
$$
\rho_1 x+\frac{1-\phi_1}{\phi_2}\rho_2 y=1.
$$
But for  $(x,y) \in \overline{D(0,1)}\times\overline{D(0,1)}$,  we have
$$
\left|\rho_1 x+\frac{1-\phi_1}{\phi_2}\rho_2 y\right| \leq \rho_1+\frac{1-\phi_1}{\phi_2}\rho_2<1
$$
under Condition~\eqref{stabcondex} and we cannot have the above equality.  This proves that the only solutions to $h_1(x,y)=h_2(x,y)=0$ in $\overline{D(0,1)}\times\overline{D(0,1)}$ are $(0,0)$ and $(1,1)$ and hence, for $y\in [y_2,1]$, we can define $P(X^*(y),0)$ by using Equation~\eqref{extPx0} for $y \in D\left(0,\sqrt{\frac{\phi_2}{\rho_2}}   \right)$.

Conversely, for $y \in \left[1,\sqrt{\frac{\rho_1}{\phi_1}}\right]$, we have
$$
P(x,0) = -\frac{h_3(x,Y^*(x))P(0,Y^*(x))+h_4(x,Y^*(x))P(0,0)}{h_2(x,Y^*(x))}.
$$
Since the function $Y^*(x)$ is analytic in the interval $\left(1,\sqrt{\frac{\rho_1}{\phi_1}}\right)$, $Y^*(x) \in (y_2,1)$, and $P(0,y)$ is analytic in $D(0,1)$, the above equation shows that the function $P(x,0)$ can analytically be continued in $\left(1,\sqrt{\frac{\phi_1}{\rho_1}}\right)$. Since $P(x,0)$ is analytic in $(0,1)$ and in $\left(1,\sqrt{\frac{\rho_1}{\phi_1}}\right)$, by Riemann's Removable Singularity Theorem \cite{Cartan}, the function $P(x,0)$ is analytic in $\left(0, \sqrt{\frac{\rho_1}{\phi_1}}\right)$. Since the function $P(x,0)$ can be expanded in power series with positive coefficients at the origin, we conclude that the function $P(x,0)$ can analytically be continued in the whole disk $D\left(0,\sqrt{\frac{\phi_1}{\rho_1}}\right)$. 

The continuity in the closed disk comes from the fact that the circle $C\left(0,\sqrt{\frac{\phi_1}{\rho_1}}\right)$, where $C\left(0,\sqrt{\frac{\phi_1}{\rho_1}}\right)$ is the circle with center 0 and radius $\sqrt{\frac{\phi_1}{\rho_1}}$, is the image by $X^\pm(y)$ of the segment $[y_1,y_2]$.

Similar arguments can be used to prove the result for $P(0,y)$.
\end{proof}

It is worth noting that in \cite{FIM}, it is shown that the function $P(x,0)$ can even be meromorphically  continued in $\C\setminus[x_3,x_4]$. We shall come up in the following with the same conclusion via a constructive method, namely by obtaining an explicit form for the function $P(x,0)$.

\subsection{Riemann-Hilbert problems}

From Equation~\eqref{extPx0} and the above proposition,  the function $P(x,0)$ is analytic in $D\left(0,\sqrt{\frac{\phi_1}{\rho_1}}\right)$ and for $x \in C\left(0,\sqrt{\frac{\phi_1}{\rho_1}}\right)$ we have
\begin{equation}
\label{RHPx}
\Re\left( i\frac{h_2(x,Y^*(x))}{h_3(x,Y^*(x))}P(x,0)\right) = \Im\left(\frac{h_4(x,Y^*(x))}{h_3(x,Y^*(x))}P(0,0)  \right).
\end{equation}
since $Y^*(x) \in [y_1,y_2]$.

Similarly,  the function $P(0,y)$ is analytic in $D\left(0,\sqrt{\frac{\phi_2}{\rho_2}}\right)$ and for $x \in C\left(0,\sqrt{\frac{\phi_2}{\rho_2}}\right)$, where $C\left(0,\sqrt{\frac{\phi_2}{\rho_2}}\right)$ is the circle with center 0 and radius $\sqrt{\frac{\phi_2}{\rho_2}}$, we have
\begin{equation}
\label{RHPy}
\Re\left( i\frac{h_3(X^*(y),y)}{h_2(X^*(y),y)}P(0,y)\right) = \Im\left(\frac{h_4(X^*(y),y)}{h_2(X^*(y),y)}P(0,0)  \right).
\end{equation}

In the case $\phi_1+\phi_2=1$ (work conserving GPS queue), the above Riemann-Hilbert problems boil down  to Dirichlet problems and have been analyzed in \cite{GP}.  In the following, we assume that $\phi_1+\phi_2>1$.

By using Equation~\eqref{h2h3h4}, we have
$$
\Im\left(\frac{h_4(x,y)}{h_2(x,y)}\right)= \frac{1-\phi_1}{1-\phi_1-\phi_2}\Re\left(i\frac{h_3(x,y)}{h_2(x,y)}\right).
$$
Equation~\eqref{RHPy} can then be rewritten as
\begin{equation}
\label{RHPyb}
\Re\left( i\frac{h_3(X^*(y),y)}{h_2(X^*(y),y)}\left(P(0,y)-\frac{1-\phi_1}{1-\phi_1-\phi_2}P(0,0)\right)\right) =0.
\end{equation}

Similarly, Equation~\eqref{RHPx} can be rewritten as
\begin{equation}
\label{RHPxb}
\Re\left( i\frac{h_2(x,Y^*(x))}{h_3(x,Y^*(x))}\left(P(x,0)-\frac{1-\phi_2}{1-\phi_1-\phi_2}P(0,0)\right)\right) =0.
\end{equation}

The above Riemann-Hilbert problems are solved in the next section.

\section{Solutions to the Riemann-Hilbert problems}
\label{rhp}

In this section, we focus on Riemann-Hilbert problem~\eqref{RHPyb}. The analysis of problem~\eqref{RHPxb} is completely symmetrical. We show in the following how the Riemann-Hilbert problem~\eqref{RHPyb} (which is of  form~\eqref{RHPgen1}) can be reduced to a problem of  form~\eqref{RHPgen2}.

In the following we use the concept of resultant of two polynomials to determine their common roots.  More precisely, if $f_1(x,y)$ and $f_2(x,y)$ are two polynomials in variables $x$ and $y$ defined by
\begin{eqnarray*}
f_1(x,y)  &=&  a_0(y) + a_1(y) x +\cdots +  a_n(y)x^n, \\ 
f_2(x,y) &=&  b_0(y) + b_1(y) x +\cdots +  b_m(y)x^m,
\end{eqnarray*}
they have a common root $y_0$ if $y_0$ is a root of their resultant $\mathrm{Res}_x(f_1,f_2)(y)$ with respect to variable $x$, which is the determinant
$$
\left| \begin{array}{cccccc}
a_n & \cdots & a_0 & 0 & \cdots & \cdots \\
0 & a_n& \cdots & a_0 & 0 & \cdots \\
\cdots & \cdots  & \cdots  & \cdots  & \cdots  & \cdots \\
\cdots & \cdots & 0 & a_n & \cdots & a_0 \\
b_m & \cdots & b_0 & 0 & \cdots & \cdots \\
0 & b_m & \cdots & b_0 & 0 & \cdots \\
\cdots & \cdots  & \cdots  & \cdots  & \cdots  & \cdots \\
\cdots & \cdots & 0 & b_m & \cdots & b_0 \\
 \end{array} \right| \begin{array}{c} \left\}\begin{array}{c}\\ m \mbox{ rows} \\  \\ \end{array} \right.   \\  \\ \left\}\begin{array}{c}\\ n \mbox{ rows}\\  \\ \end{array} \right.\end{array}
$$
Note that this resultant can also be expressed as a combination $p(x,y)f_1(x,y) +q(x,y)f_2(x,y)$, where $p$ and $q$ are polynomials in variables $x$ and $y$. 

The function $P(0,y)$ is analytic in the open disk $D\left(0,\sqrt{\frac{\phi_2}{\rho_2}}\right)$. By using the reflection principle \cite{Cartan}, the function 
$$
y \mapsto  \overline{P\left(0, \frac{\phi_2}{\rho_2\overline{y}}  \right)}
$$
is analytic on the outside of the closed disk $\overline{D\left(0,\sqrt{\frac{\phi_2}{\rho_2}}\right)}$. It is then easily checked that if we define 
\begin{equation}
\label{FYP0y}
F_Y(y) = \left\{ \begin{array}{ll}
P(0,y)-\frac{1-\phi_1}{1-\phi_1-\phi_2}P(0,0), & y \in D\left(0,\sqrt{\frac{\phi_2}{\rho_2}}\right),\\ \\
 \overline{P\left(0, \frac{\phi_2}{\rho_2\overline{y}}  \right)} -\frac{1-\phi_1}{1-\phi_1-\phi_2}P(0,0) & y \in \C\setminus \overline{D\left(0,\sqrt{\frac{\phi_2}{\rho_2}}\right)},
\end{array}\right.
\end{equation}
the function $F_Y(y)$ is sectionally analytic with respect to the circle $C\left(0,\sqrt{\frac{\phi_2}{\rho_2}}\right)$, $F_Y(y)$ tends to $\frac{-\phi_2}{1-\phi_1-\phi_2}P(0,0)$ when $y \to \infty$, and for $y \in C\left(0,\sqrt{\frac{\phi_2}{\rho_2}}\right)$
\begin{equation}
\label{Hilberty}
F^i_Y(y) = \alpha_Y(y)F^e_Y(y),
\end{equation}
where $F^i_Y(y)$ (resp. $F^e_Y(y)$) is the interior (resp. exterior) limit of the function $F_Y(y)$ at the circle $C\left(0,\sqrt{\frac{\phi_2}{\rho_2}}\right)$, and the function $\alpha_Y(y)$ is defined on $C\left(0,\sqrt{\frac{\phi_2}{\rho_2}}\right)$ by
\begin{equation}
\label{defalphayini}
\alpha_Y = \frac{\overline{a_Y(y)}}{a_Y(y)}
\end{equation}
with
$$
a_Y(y) = \frac{h_3(X^*(y),y)}{h_2(X^*(y),y)}.
$$
Problem~\eqref{Hilberty} is clearly of form~\eqref{RHPgen2}.

If we show that Problem~\eqref{Hilberty} has a unique solution $F_Y(y)$, then $P(0,y)$ and $F_Y(y)$ are related to each other according to Equation~\eqref{FYP0y} in the disk $D\left(0,\sqrt{\frac{\phi_2}{\rho_2}}\right)$.

The generic solutions to the Riemann-Hilbert problems of form~\eqref{RHPgen2}  are given in \cite{Lions}. We first have to determine the index of the problem defined as
$$
\kappa_Y = \frac{1}{2\pi}\var_{y \in C\left(0,\sqrt{\frac{\phi_2}{\rho_2}}\right)} \arg \alpha_Y(y)  .
$$

Let us first study the function $a_Y(y)$, which  can be expressed as follows. 

\begin{lemma}
The function  $\alpha_Y(y)$ defined for $y \in C\left(0,\sqrt{\frac{\phi_2}{\rho_2}}\right)$ by Equation~\eqref{defalphayini} can be extended as a meromorphic function in $\C\setminus([y_1,y_2]\cup[y_3,y_4])$ bet setting
\begin{equation}
\label{defalphaY}
\alpha_Y(y) = \frac{-\mu_1(1-\phi_1-\phi_2)\phi_2\mu_2X^*(y) + yR_Y(X^*(y))}{ y(-\lambda_2\mu_1X^*(y)(1-\phi_1-\phi_2)y + R_Y(X^*(y))) },
\end{equation}
where
\begin{multline}
\label{defRY}
R_Y(x) = (1-\phi_2)\lambda_1(\phi_2\mu_2-(1-\phi_1)\mu_1)x^2 \\ +\left((1-\phi_1)(1-\phi_2)(\lambda_1+\lambda_2)-  \phi_1(\phi_2\mu_2-(1-\phi_1)\mu_1)\right)\mu_1x -\phi_1(1-\phi_1)\mu_1^2.
\end{multline}
\end{lemma}

\begin{proof}
Let us introduce the resultant $\mathrm{Res}_y(h_1,h_2)(x)$ of the polynomials $h_1(x,y)$ and $h_2(x,y)$, which is given by the determinant
$$
\mathrm{Res}_y(h_1,h_2)(x) = \left|
\begin{array}{ccc}
-\lambda_2 x &\alpha_2(x) & -\phi_2\mu_2x \\
 \beta_2(x) & -\phi_2\mu_2x &0\\
0 &\beta_2(y) & -\phi_2\mu_2x
\end{array}
\right|,
$$
where 
\begin{eqnarray*}
\alpha_2(x) &=&  -(\lambda_1x^2-(\lambda_1+\lambda_2+\phi_1\mu_1+\phi_2\mu_2)x+\phi_1\mu_1)\\
\beta_2(x) &=&\phi_2\mu_2x +(1-\phi_1)\mu_1(1-x) .
\end{eqnarray*}

The  polynomial $\mathrm{Res}_y(h_1,h_2)(x)$ can be written as
\begin{equation}
\label{resyh1h2x}
\mathrm{Res}_y(h_1,h_2)(x) = q_Y(x,y)h_2(x,y)
\end{equation}
where $q_Y(x,y)$ is a polynomial in $x$ and $y$. It follows that 
$$
a_Y(y) = \frac{q_Y(x,y)h_3(x,y)}{\mathrm{Res}_y(h_1,h_2)(x)}.
$$
Since for $y \in C\left(0,\sqrt{\frac{\phi_2}{\rho_2}}\right)$ and $x=X^*(y) \in [x_1,x_2]$, we have
$$
\alpha_Y(y) = \frac{\overline{q_Y(x,y)h_3(x,y)}}{q_Y(x,y)h_3(x,y)}.
$$

Simple computations show that 
$$
q_Y(x,y) = \lambda_2xy\beta_2(x)+\lambda_2\phi_2\mu_2x^2-\alpha_2(x)\beta_2(x) .
$$

Writing $h_3(x,y)  = (1-\phi_2)\mu_2 x -\beta_3(x)y $ with $\beta_3(x) = (1-\phi_2)\mu_2 x+\phi_1\mu_1(1-x)$, we have for $y \in C\left(0,\sqrt{\frac{\phi_2}{\rho_2}}\right)$ and $x=X^*(y)$
\begin{multline*}
q_Y(x,y) h_3(x,y)= -\lambda_2\beta_2(x)\beta_3(x)xy^2\\ +\left((1-\phi_2)\mu_2\lambda_2\beta_2(x)x^2+\alpha_2(x)\beta_2(x)\beta_3(x)-\lambda_2\phi_2\mu_2\beta_3(x)x^2  \right)y  \\ +(1-\phi_2)\mu_2x(\lambda_2\phi_2\mu_2x^2-\alpha_2(x)\beta_2(x)).
\end{multline*}
By using the fact that $\lambda_2xy^2 = \alpha_2(x)y-\mu_2\phi_2x$, we obtain
\begin{multline*}
q_Y(x,y) h_3(x,y)= \lambda_2\mu_2x^2\left((1-\phi_2)\beta_2(x) - \phi_2\beta_3(x) \right)y  \\ +\phi_2\mu_2\beta_2(x)\beta_3(x)x +(1-\phi_2)\mu_2x(\lambda_2\phi_2\mu_2x^2-\alpha_2(x)\beta_2(x)).
\end{multline*}
Simple computations show that 
$$
(1-\phi_2)\beta_2(x) - \phi_2\beta_3(x) =(1-\phi_1-\phi_2)\mu_1(1-x).
$$
In addition,
$$
\phi_2\mu_2\beta_2(x)\beta_3(x)x +(1-\phi_2)\mu_2x(\lambda_2\phi_2\mu_2x^2-\alpha_2(x)\beta_2(x))= \mu_2x(x-1) R_Y(x),
$$
where the polynomial $R_Y(x)$ is defined by Equation~\eqref{defRY}. 

It follows that
$$
q_Y(x,y) h_3(x,y) = \mu_2x(x-1)\left(-\lambda_2\mu_1x(1-\phi_1-\phi_2)y + R_Y(x)\right)
$$
and Equation~\eqref{defalphaY} follows.
\end{proof}

By using the above lemma, we can now determine the index of the Riemann-Hilbert problem~\eqref{Hilberty}.

\begin{lemma}
Under Condition~\eqref{stabcondex}, the index of the Riemann-Hilbert problem~\eqref{Hilberty} is $\kappa_Y = 0$.
\end{lemma}

\begin{proof}
For $y \in C\left(0,\sqrt{\frac{\phi_2}{\rho_2}}\right)$ and $x = X^*(y)$, let the function $f_Y(x)$ be defined by 
\begin{multline*}
 f_Y(x) = -\lambda_2\mu_1x(1-\phi_1-\phi_2)\Re(y) + R_Y(x) \\
=  \frac{1}{2}(1-\phi_1-\phi_2)\mu_1\left(\lambda_1x^2-(\lambda_1+\lambda_2+\phi_1\mu_1+\phi_2\mu_2)x+\phi_1\mu_1\right) +R_Y(x).
\end{multline*}
The function $f_Y(x)$ is a  quadratic polynomial such that 
$$
f_Y(0) =  \frac{1}{2}(1-\phi_1-\phi_2)\phi_1\mu_1^2-\phi_1(1-\phi_1)\mu_1^2<0
$$
and
\begin{multline*}
f_Y(1) = \\   \mu_1\mu_2\left((1-\phi_2)\phi_2\rho_1+(1-\phi_1)(1-\phi_2)\rho_2-\phi_1\phi_2-\frac{(1-\phi_1-\phi_2)(\phi_2+\rho_2)}{2}\right).
\end{multline*}

If $\rho_2 \geq\phi_2$
\begin{multline*}
f_Y(1) \leq \\  \mu_1\mu_2\left((1-\phi_2)\phi_2\rho_1+(1-\phi_1)(1-\phi_2)\rho_2-\phi_1\phi_2-(1-\phi_1-\phi_2)\rho_2 \right)\\
\leq  \mu_1\mu_2\phi_1\phi_2\left(\frac{1-\phi_2}{\phi_1}\rho_1+\rho_2-1  \right)<0
\end{multline*}
under Condition~\eqref{stabcondex}. If $\rho_2<\phi_2$,
\begin{multline*}
f_Y(1) \leq \\   \mu_1\mu_2\left((1-\phi_2)\phi_2\rho_1+(1-\phi_1)(1-\phi_2)\rho_2-\phi_1\phi_2-(1-\phi_1-\phi_2)\phi_2 \right)\\
\leq  \mu_1\mu_2(1-\phi_2)\phi_2\left(\rho_1+\frac{1-\phi_1}{\phi_2}\rho_2-1  \right)<0
\end{multline*}
under Condition~\eqref{stabcondex}. We hence deduce that we always have $f_Y(1)<0$  under Condition~\eqref{stabcondex}.

The coefficient of the leading term of the polynomial $f_Y(x)$ is $c_Y$ given by
$$
c_Y =\left(\frac{1}{2}(1-\phi_1-\phi_2)\mu_1+(1-\phi_2)(\phi_1\mu_2-(1-\phi_1)\mu_1)\right).
$$

If $c_Y>0$ then the polynomial $f_Y(x)$ has two roots with opposite signs since $f_Y(0)<0$. The fact that $f_Y(1)<0$ implies that $f_Y(x)<0$ for all $x$ in $[0,1]$. If $c_Y<0$, then either the quadratic polynomial has no roots and is always negative or else has two roots. If the roots exist, then they are either both positive or both negative. If they are negative, then $f_Y(x)$ is negative for all $x\in [0,1]$. If the roots are positive, then they are either both in $(0,1)$ or else greater than 1 since $f_Y(1)<0$. The product of the roots is equal to 
\begin{multline*}
\frac{(\phi_1(1-\phi_1)+\frac{1}{2}(\phi_1+\phi_2-1)\phi_1)\mu_1^2}{(((1-\phi_1)(1-\phi_2)\mu_1+ \frac{1}{2}(\phi_1+\phi_2-1)\phi_1)\mu_1-(1-\phi_2)\phi_2\mu_2)\lambda_1}> \\\frac{1}{\rho_1}\frac{\phi_1(1-\phi_1)+\frac{1}{2}(\phi_1+\phi_2-1)\phi_1}{(1-\phi_1)(1-\phi_2)+ \frac{1}{2}(\phi_1+\phi_2-1)\phi_1}>  \\ \frac{\phi_1(1-\phi_1)+\frac{1}{2}(\phi_1+\phi_2-1)\phi_1}{(1-\phi_1)(1-\phi_2)+ \frac{1}{2}(\phi_1+\phi_2-1)\phi_1}>1
\end{multline*}
since $\rho_1<1$ under Condition~\eqref{stabcondex} and $\phi_1+\phi_2>1$. Hence, under the condition $c_Y<0$, if the roots of the polynomial $f_Y(x)$ exist and are positive, they are both greater than 1. This shows that in all cases $f_Y(x)$ is negative for $x \in [0,1]$.

It then follows that when $y$ traverses the circle $C(0,\sqrt{\phi_2/\rho_2})$,   the quantity $-\lambda_2\mu_1X^*(y)(1-\phi_1-\phi_2)\Re(y) + R_Y(X^*(y)) \leq 0$ and then the closed contour described by $-\lambda_2\mu_1X^*(y)(1-\phi_1-\phi_2)y + R_Y(X^*(y))$ entirely lies in the half plane $\{z : \Re(z)<0\}$. This implies  that $\kappa_Y = 0$.
\end{proof}

Since the index of the Riemann-Hilbert~\eqref{Hilberty} is null, the solution is as follows.

\begin{lemma}
\label{lemHilbert}
The solution to the Riemann-Hilbert problem~\eqref{Hilberty} exists and is unique and  given for $y \in \C \setminus C\left(0,\sqrt{\frac{\phi_2}{\rho_2}}\right)$ by
\begin{equation}
\label{defFy}
F_Y(y) =  \frac{-\phi_2}{1-\phi_1-\phi_2} P(0,0)  \varphi_Y(y),
\end{equation}
where
\begin{equation}
\label{defvarphiy}
\varphi_Y(y) = \exp\left( \frac{y}{\pi} \int_{x_1}^{x_2}\frac{(\lambda_1x^2-\phi_1\mu_1)\Theta_Y(x)}{xh_1(x,y)}dx  \right)
\end{equation}
and
\begin{multline}
\label{defThetaY}
\Theta_Y(x) = \\  \atan\left(\frac{\mu_1 (\phi_1+\phi_2-1)\sqrt{-D_1(x)}}{\mu_1(\phi_1+\phi_2-1)(\lambda_1x^2-(\lambda_1+\lambda_2+\phi_1\mu_1+\phi_2\mu_2)x+\phi_1\mu_1)-2R_Y(x)}  \right) .
\end{multline}
\end{lemma}

\begin{proof}
Since under Condition~\eqref{stabcondex}, the index of the Riemann-Hilbert~\eqref{Hilberty} is null, the solution reads
$$
F_Y(y) = \phi_Y(y) \exp\left(\frac{1}{2i\pi}\int_{C\left(0,\sqrt{\frac{\phi_2}{\rho_2}}\right)} \frac{\log\alpha_Y(z)}{z-y} dz  \right)
$$
where the function $\alpha_Y(y)$ is defined by Equation~\eqref{defalphaY} and $\phi_Y(y)$ is a polynomial. Since we know that $F_Y(y)\to -\phi_2P(0,0)/(1-\phi_1-\phi_2)$ as $|y|\to \infty$, then
$$
\phi_Y(y) =  \frac{-\phi_2}{1-\phi_1-\phi_2}P(0,0).
$$

Let for $y \in C\left(0,\sqrt{\frac{\phi_2}{\rho_2}}\right)$ and $y = Y^*(x+i0)$ for $x\in [x_1,x_2]$
$$
\Theta_Y(x) = \arg\left(\lambda_2\mu_1 x (1-\phi_1-\phi_2)Y^*(x+0i) - R_Y(x)\right)
$$
By using the expression of $Y^*(x)$, Equation~\eqref{defThetaY} follows. It is clear that  
$$
\log\alpha_Y(Y^*(x+0i)) = -2i\Theta_Y(x).
$$

Since $Y^*(x+0i) = \overline{Y^*(x-0i)}$, we have
\begin{align*}
\frac{1}{2i\pi}\int_{C\left(0,\sqrt{\frac{\phi_2}{\rho_2}}\right)} \frac{\log\alpha_Y(z)}{z-y} dz &=\frac{1}{\pi} \int_{x_1}^{x_2}\Im\left(\frac{\log\alpha_Y(Y^*(x+0i))}{Y^*(x+0i)-y}\frac{dY^*}{dx}(x+0i) \right)dx \\
&=\frac{1}{\pi} \int_{x_1}^{x_2}\Im\left(\frac{-2i}{Y^*(x+0i)-y}\frac{dY^*}{dx}(x+0i) \right)\Theta_Y(x)dx
\end{align*}
It is easily checked from the equation $h_1(x,Y^*(x) )=0$ that
$$
\frac{dY^*}{dx} = \frac{-2\lambda_1xY^*(x)-\lambda_2Y^*(x)^2+(\lambda_1+\lambda_2+\phi_1\mu_1+\phi_2\mu_2)Y^*(x) -\phi_2\mu_2}{\lambda_1x^2+2\lambda_2 xY^*(x)-(\lambda_1+\lambda_2+\phi_1\mu_1+\phi_2\mu_2)x+\phi_1\mu_1}
$$
For $x \in [x_1,x_2]$, we have
$$
\lambda_1x^2+2\lambda_2 xY^*(x+0i)-(\lambda_1+\lambda_2+\phi_1\mu_1+\phi_2\mu_2)x+\phi_1\mu_1=-i\sqrt{-D_1(x)}
$$
By using once again $h_1(x,Y^*(x+0i))=0$, we obtain for $x \in [x_1,x_2]$
$$
\frac{dY^*}{dx}(x+0i) =\frac{(\phi_1\mu_1-\lambda_1x^2)Y^*(x+0i)}{-ix\sqrt{-D_1(x)}}
$$
and then for real $y$
$$
\Im\left(\frac{-2i}{Y^*(x+0i)-y}\frac{dY^*}{dx}(x+0i) \right) = \frac{(\lambda_1x^2-\phi_1\mu_1)y}{xh_1(x,y)}.
$$
It follows that for real $y$
$$
\frac{1}{2i\pi}\int_{C\left(0,\sqrt{\frac{\phi_2}{\rho_2}}\right)} \frac{\log\alpha_Y(z)}{z-y} dz = \frac{y}{\pi} \int_{x_1}^{x_2}\frac{(\lambda_1x^2-\phi_1\mu_1)\Theta_Y(x)}{xh_1(x,y)}dx
$$
It is easily checked that the function on the right hand side of the above equation can analytically be continued in the disk $D\left(0,\sqrt{\frac{\phi_2}{\rho_2}}\right)$. Hence for $y \in D\left(0,\sqrt{\frac{\phi_2}{\rho_2}}\right)$, the first part of  Equation~\eqref{defFy} follows. When $y$ is not in the closed disk $\overline{D\left(0,\sqrt{\frac{\phi_2}{\rho_2}}\right)}$, similar arguments can be used to derive the second part of Equation~\eqref{defFy}.\end{proof}

In view of the above lemma, we can state the main result of this section.

\begin{theorem}
The function $P(0,y)$ can be defined as a meromorphic function in $\C\setminus[y_3,y_4]$ by setting 
\begin{equation}
\label{defP0y}
 {P(0,y)} =\left\{  \begin{array}{l}
\frac{-\phi_2P(0,0)}{1-\phi_1-\phi_2}  \varphi_Y(y)+ \frac{(1-\phi_1)P(0,0)}{1-\phi_1-\phi_2}, \; y \in D\left(0,\sqrt{\frac{\phi_2}{\rho_2} } \right), \\ \\ 
\frac{-\phi_2P(0,0)}{1-\phi_1-\phi_2}\alpha_Y(y) \varphi_Y(y) + \frac{(1-\phi_1)P(0,0)}{1-\phi_1-\phi_2}, \; y \notin D\left(0,\sqrt{\frac{\phi_2}{\rho_2} } \right),
\end{array}\right.
\end{equation}
where $\varphi_Y(y)$ is defined by Equation~\eqref{defvarphiy}.
\end{theorem}

\begin{proof}Since the solution to the Riemann-Hilbert problem~\eqref{RHPyb} is unique, the function $P(0,y)$ coincides with the function $F_Y(y)+(1-\phi_1)P(0,0)/(1-\phi_1-\phi_2)$ in $D\left(0,\sqrt{\frac{\phi_2}{\rho_2} } \right)$. We can extend this function as follows (see \cite{Nova} for details). Noting that the function $\log\alpha_Y(y)$ is analytic in a neighborhood of the circle $C\left(0,\sqrt{\frac{\phi_2}{\rho_2} } \right)$, the function 
$$
y \mapsto \exp\left( \frac{1}{2i\pi}\int_{C\left(0,\sqrt{\frac{\phi_2}{\rho_2}}\right)} \frac{\log\alpha_Y(z)}{z-y} dz \right)
$$
defined for $y \in D\left(0,\sqrt{\frac{\phi_2}{\rho_2} } \right)$ can be continued as a meromorphic function in $\C\setminus[x_3,x_4]$ by considering the function defined for $y \notin \overline{D\left(0,\sqrt{\frac{\phi_2}{\rho_2} } \right)}$ by
\begin{multline*}
\alpha_Y(y) \exp\left(\frac{1}{2i\pi}\int_{C\left(0,\sqrt{\frac{\phi_2}{\rho_2}}\right)} \frac{\log\alpha_Y(z)}{z-y} dz \right) = \\  \alpha_Y(y) \exp\left( \frac{y}{\pi} \int_{x_1}^{x_2}\frac{(\lambda_1x^2-\phi_1\mu_1)\Theta_Y(x)}{xh_1(x,y)}dx  \right),
\end{multline*}
where the last equality is obtained by using the same arguments as above (consider first real $y$ and then extend the function by analytic continuation).
\end{proof}

The poles of the function $P(0,y)$ are the poles of the function $\alpha_Y(y)$ defined by Equation~\eqref{defalphaY}, which can be rewritten as
$$
\alpha_Y(y) = \frac{h_3\left(X^*(y),\frac{\phi_2}{\rho_2 y}\right) h_2(X^*(y),y)}{h_2\left(X^*(y),\frac{\phi_2}{\rho_2 y}\right) h_3(X^*(y),y)}.
$$
The poles of the function $\alpha_Y(y)$ are clearly  related to the solutions to the equations $h_3(X^*(y),y)=0$ and $h_2\left(X^*(y),\frac{\phi_2}{\rho_2 y}\right) =0$. This observation leads us to introduce the resultants with respect to $x$ of the polynomials $h_1(x,y)$ and $h_2(x,y)$ on the one hand and $h_1(x,y)$ and $h_3(x,y)$ on the other hand.

The resultant $\mathrm{Res}_x(h_1,h_2)(y) = \mu_1y(y-1)P_X(y)$, where the polynomial $P_X(y)$ is given by 
\begin{multline}
\label{defPX}
{P}_{X}(y) = \lambda_2(1-\phi_1)(\phi_2\mu_2-(1-\phi_1)\mu_1)y^2 \\ -\phi_2\mu_2((1-\phi_1)(\lambda_1+\lambda_2)-\mu_1(1-\phi_1)+\mu_2\phi_2)y  +\phi_2^2\mu_2^2.
\end{multline}
The roots of this polynomial are given by
\begin{equation}
\label{ypm}
y_\pm = \phi_2   \frac{((1-\phi_1)(\lambda_1+\lambda_2)-\mu_1(1-\phi_1)+\mu_2\phi_2) \pm\sqrt{\Delta_X^{(1)}}}{2\rho_2 (1-\phi_1)(\phi_2\mu_2-(1-\phi_1)\mu_1)},
\end{equation}
where
\begin{multline*}
\Delta_X^{(1)} = ((1-\phi_1)(\lambda_1+\lambda_2)-\mu_1(1-\phi_1)+\mu_2\phi_2)^2-4 \lambda_2(1-\phi_1)(\phi_2\mu_2-(1-\phi_1)\mu_1).
\end{multline*}

The resultant $\mathrm{Res}_x(h_1,h_3)(y) = -\phi_1\mu_1y(y-1)Q_X(y)$, where
\begin{multline}
Q_X(y) = \lambda_2(\phi_1\mu_1-(1-\phi_2)\mu_2)y^2 \\ +\left((1-\phi_2)(\lambda_1+\lambda_2)-(\phi_1\mu_1-(1-\phi_2)\mu_2)  \right)\mu_2y-(1-\phi_2)\mu_2^2.
\end{multline}
 The polynomial $Q_X(y)$ has two real roots given by
\begin{equation}
\label{xipm}
\xi_\pm = \frac{-\left((1-\phi_2)(\lambda_1+\lambda_2)-(\phi_1\mu_1-(1-\phi_2)\mu_2)  \right) \pm \sqrt{\Delta^{(2)}_Y}}{2\rho_2 (\phi_1\mu_1-(1-\phi_2)\mu_2)},
\end{equation}
where
\begin{multline*}
\Delta^{(2)}_Y =  \\ \left((1-\phi_2)(\lambda_1+\lambda_2)-(\phi_1\mu_1-(1-\phi_2)\mu_2)^2+4(1-\phi_2) \lambda_2(\phi_1\mu_1-(1-\phi_2)\mu_2)    \right).
\end{multline*}
If $\phi_1\mu_1-(1-\phi_2)\mu_2\leq0$, $\xi_\pm$ are both positive. If $\phi_1\mu_1-(1-\phi_2)\mu_2>0$, $\xi_-$ is negative and $\xi_+$ is positive. The solution $\xi_+$ is always the positive root with the smallest module.

By using Lemmas~\ref{lemmah1h2} and \ref{lemmah1h3} proved in Appendix, we can show the following result.

\begin{proposition}
\label{P0yresult}
The function $P(0,y)$ is analytic in the disk $D(0,\rho_Y)$ with center 0 and radius $\rho_Y$ given by
$$
\rho_Y = \left\{
\begin{array}{ll}
\xi_+ & \mbox{if} \; \phi_1>\rho_1 \; \mbox{and} \; Q_Y\left(\sqrt{\frac{\phi_1}{\rho_1}}  \right)<0, \\
y_3 & \mbox{otherwise}.
\end{array}
\right.
$$
The function $P(0,y)$ satisfies Equation~\eqref{RHPy}.
\end{proposition}

\begin{proof}
The radius of convergence $\rho_Y$ can easily be deduced from  Lemmas~\ref{lemmah1h2} and \ref{lemmah1h3} proved in Appendix
Since $\rho_Y>\sqrt{\frac{\phi_2}{\rho_2}}$, this implies that the series expansion $\sum_{n=0}^\infty p(0,n)y^n$ valid in $D(0,1)$ by definition is also valid in $D\left( 0,\rho_Y\right)$. This implies that if $y$ is sufficiently close to the circle $C\left( 0, \sqrt{\frac{\phi_2}{\rho_2}}\right)$ on the inside, we have
$$
P(0,y) = \sum_{n=0}^\infty p(0,n)y^n \equiv F^i_Y(y)+\frac{1-\phi_1}{1-\phi_1-\phi_2}P(0,0).
$$
The point $\phi_2/(\rho_2\overline{y})$ is close to $y$ but on the outside of the disk $D\left( 0, \sqrt{\frac{\phi_2}{\rho_2}}\right)$ and we have
$$
\overline{P\left(0,\frac{\phi_2}{\rho_2\overline{y}}\right)} = \sum_{n=0}^\infty p(0,n)\overline{y}^n\equiv F^e_Y(y)+\frac{1-\phi_1}{1-\phi_1-\phi_2}P(0,0).
$$
By using the fact that the function $F_Y(y)$ satisfies Equation~\eqref{Hilberty}, we immediately deduce that the function $P(0,y)$ satisfies Equation~\eqref{RHPy}.
\end{proof}

To conclude this section, let us  give the value of  $P(0,0)$, which is different from the classical $(1-\rho)$-formula valid for work conserving systems.

\begin{corollary}
The quantity $P(0,0)$ is given by
\begin{equation}
P(0,0) = \left\{
\begin{array}{ll}
\frac{1}{\varphi_Y(1)}\left(1-\rho_1-\frac{1-\phi_1}{\phi_2}\rho_2\right) & \mbox{if} \; \phi_2> \rho_2,\\
\frac{\phi_1}{(1-\phi_2)\varphi_Y(1)}\left(1-\frac{1-\phi_2}{\phi_1}\rho_1-\rho_2\right) & \mbox{if} \; \phi_2\leq \rho_2.
\end{array}
\right.
\end{equation}
\end{corollary}

\begin{proof}
In the case $\phi_2> \rho_2$, we have by using Equations~\eqref{P01} and \eqref{defP0y}
\begin{equation}
P(0,0) = \frac{1}{\varphi_Y(1)}\left(1-\rho_1-\frac{1-\phi_1}{\phi_2}\rho_2\right).
\end{equation}
When $\phi_2\leq \rho_2$ and then $\rho_1<\phi_1$ and $X^*(1)=1$, we note that 
$$
R_Y(1)= \mu_1\mu_2\left(\phi_2(1-\phi_2)\rho_1+(1-\phi_1)(1-\phi_2)\rho_2-\phi_1\phi_2  \right) 
$$
and then
$$
\alpha_Y(1) = \frac{1-\phi_2}{\phi_1}\frac{1-\rho_1-\frac{1-\phi_1}{\phi_2}\rho_2}{1-\frac{1-\phi_2}{\phi_1}\rho_1-\rho_2},
$$
and we deduce that
$$
P(0,0)= \frac{\phi_1}{(1-\phi_2)\varphi_Y(1)}\left(1-\frac{1-\phi_2}{\phi_1}\rho_1-\rho_2\right).
$$
\end{proof}

The computation of the quantity $\varphi_Y(1)$ involves elliptic integrals but can easily be performed by using Computer Algebra Systems such as Mathematica.

A result similar to Proposition~\ref{P0yresult} holds for function $P(x,0)$. By using in addition, the explicit form of $P(0,0)$ given in the above corollary, we have completely determined the generating function $P(x,y)$.

\section{Asymptotic analysis}
\label{asymptotics}

In this section, we investigate the tail asymptotics of the distribution of the number $N_2$ of jobs  in queue $\# 2$, whose probability generating  function is given by Equation~\eqref{P1y}. It clearly appears from this equation that the point $y = \phi_2/\rho_2$ is a potential pole for that function. The following lemma states in which conditions the point $\phi_2/\rho_2$ is  a  removable singularity for the function $P(1,y)$.

\begin{lemma}
\label{removepole}
If $\phi_2 \leq \rho_2$ (and then necessarily $\phi_1>\rho_1$ under Condition~\eqref{stabcondex}) or if    $\phi_2>\rho_2$ and $\phi_1\geq\rho_1$, then the point $\phi_2/\rho_2$ is a removable singularity for the function $P(1,y)$.
If $\phi_2>\rho_2$ and $\phi_1<\rho_1$
$$
\phi_2P(1,0)-(1-\phi_2)P\left(0,\frac{\phi_2}{\rho_2}  \right)+(1-\phi_2)P(0,0)>0
$$
and the point $\phi_2/\rho_2$ is a pole for the function $P(1,y)$.
\end{lemma}

\begin{proof}
If $\phi_2\leq \rho_2$, the point $\phi_2/\rho_2$ is a removable singularity for the function $P(1,y)$ because this function shall be  analytic in the closed unit disk.

If $\phi_2>\rho_2$ and $\phi_1\geq \rho_1$, we have $X^*\left(\frac{\phi_2}{\rho_2}  \right)=1$ and from the fundamental equation~\eqref{eqfund}, we deduce that
$$
h_2\left(1,\frac{\phi_2}{\rho_2}\right)P(1,0)+h_3\left(1,\frac{\phi_2}{\rho_2}  \right)P\left(0,\frac{\phi_2}{\rho_2}\right)+h_4\left(1,\frac{\rho_2}{\phi_2}\right)P(0,0)=0,
$$
which implies that 
$$
\phi_2 P(1,0)-(1-\phi_2)P\left(0,\frac{\phi_2}{\rho_2}\right)+(1-\phi_2)P(0,0)=0,
$$
and hence that the point $\phi_2/\rho_2$ is a removable singularity for the function $P(0,y)$.

If $\phi_2>\rho_2$ and $\phi_1< \rho_1$, then $X^*\left(\frac{\phi_2}{\rho_2}  \right)=\frac{\phi_1}{\rho_1}$ and the point $\phi_1/\rho_1$ has to be a removable singularity for the function $P(x,0)$ which implies from Equation~\eqref{Px1} that
$$
(1-\phi_1)P\left(\frac{\phi_1}{\rho_1},0 \right)-\phi_1P(0,1)-(1-\phi_1)P(0,0) = 0.
$$
In addition, Equation~\eqref{eqfund} for the point  $\left(\frac{\phi_1}{\rho_1},\frac{\phi_2}{\rho_2} \right)$ yields
$$
h_2\left(\frac{\phi_1}{\rho_1},\frac{\phi_2}{\rho_2}\right)P\left(\frac{\phi_1}{\rho_1},0\right)+h_3\left(\frac{\phi_1}{\rho_1},\frac{\phi_2}{\rho_2}  \right)P\left(0,\frac{\phi_2}{\rho_2}\right)+h_4\left(\frac{\phi_1}{\rho_1},\frac{\phi_2}{\rho_2}\right)P(0,0)=0.
$$
By combining the two above equations and Equations~\eqref{P10} and \eqref{P01}, we obtain after some algebra
\begin{multline}
\label{residuP0y}
\phi_2 P(1,0)-(1-\phi_2)P\left(0,\frac{\phi_2}{\rho_2}\right)+(1-\phi_2)P(0,0)= \\ \frac{\phi_2(\rho_1-\phi_1)(\phi_2-\rho_2)\left((1-\phi_1)\mu_1+(1-\phi_2)\mu_2  \right)}{(1-\phi_1)\left(\phi_2(\rho_1-\phi_1)\mu_1+(1-\phi_2)(\phi_2-\rho_2)\mu_2\right)}>0,
\end{multline}
which completes the proof.
\end{proof}

From Equation~\eqref{P1y}, we observe that the poles of the function $P(0,y)$ can  also be potential poles for the function $P(1,y)$, which can be located by using Lemmas~\ref{lemmah1h2} and \ref{lemmah1h3}. In addition, to state the asymptotic results, let us introduce the resultants with respect to $y$. The resultant $\mathrm{Res}_y(h_1,h_2)(x)$ of the polynomials $h_1(x,y)$ and $h_2(x,y)$ with respect to $y$ is equal to $-\phi_2\mu_2 x(x-1)P_Y(x)$ with 
\begin{multline}
\label{defPY}
P_Y(x) = \lambda_1(\phi_2\mu_2-(1-\phi_1)\mu_1)x^2\\ +((1-\phi_1)(\lambda_1+\lambda_2)-(\phi_2\mu_2-(1-\phi_1)\mu_1))\mu_1x-(1-\phi_1)\mu_1^2.
\end{multline}

The resultant $\mathrm{Res}_y(h_1,h_3)(x) =   \mu_2x(x-1)Q_Y(x)$, where the polynomial $Q_Y(x)$ is given  by
\begin{multline}
\label{defQY}
{Q}_{Y}(x) = \lambda_1(1-\phi_2)(\phi_1\mu_1-(1-\phi_2)\mu_2)x^2 \\ -\phi_1\mu_1((1-\phi_2)(\lambda_2+\lambda_1)-\mu_2(1-\phi_2)+\mu_1\phi_1)x  +\phi_1^2\mu_1^2.
\end{multline}

By using Lemmas~\ref{removepole}, \ref{lemmah1h2} and \ref{lemmah1h3}, we can state the following result for the tail of the distribution of the number of jobs  in queue \#2.
\begin{proposition}
\label{asymptoticest}
The tail of the probability distribution function of the number $N_2$ of jobs in queue $\#2$ is given when $n \to \infty$ by:
\begin{itemize}
\item[(a)] If  $\phi_1>\rho_1$  and $Q_Y(\sqrt{\frac{\phi_1}{\rho_1}})<0$, 
\begin{equation}
\label{equiva}
P(N_2 = n) \sim \frac{-\phi_2(1-\phi_2)P(0,0)r(\xi_+)\varphi_Y(\xi_+)}{ \xi_+(1-\phi_1-\phi_2)(\phi_2-\rho_2\xi_+)} \left( \frac{1}{\xi_+} \right)^n ,
\end{equation}
where 
$$
r(\xi_+) = \frac{ h_3\left(X^*(\xi_+),\frac{\phi_2}{\rho_2 \xi_+}  \right)h_2\left(X^*(\xi_+),\xi_+ \right)}{h_2\left(X^*(\xi_+),\frac{\phi_2}{\rho_2 \xi_+}  \right)\left( \frac{\partial  h_3}{\partial x}\left(X^*(\xi_+),\xi_+ \right)\frac{d X^*}{dx}(\xi_+) + \frac{\partial  h_3}{\partial y}\left(X^*(\xi_+),\xi_+ \right)      \right)};
$$
\item[(b)] If  $\phi_1>\rho_1$  and $Q_Y(\sqrt{\frac{\phi_1}{\rho_1}})=0$, 
\begin{equation}
\label{equivb}
P(N_2 = n) \sim   \frac{\kappa  \left(\rho_2 y_3^2 -\phi_2  \right)\varphi_Y(y_3)\sqrt{y_3(y_3-y_1)(y_3-y_2)(y_4-y_3)}}{2\sqrt{\pi}y^3_3 \rho_2(\rho_2y_3-\phi_2)Q'_X(y_3)P_X\left(\frac{\phi_2}{\rho_2 y_3}\right)} \frac{1}{\sqrt{n} y_3^n} \;;
\end{equation}
\item[(c)] If  $\phi_1>\rho_1$  and $Q_Y(\sqrt{\frac{\phi_1}{\rho_1}})>0$, 
\begin{equation} \label{equivc}
P(N_2 = n)  \sim  \frac{\kappa  \left(\rho_2 y_3^2 -\phi_2  \right)\varphi_Y(y_3)\sqrt{y_3(y_3-y_1)(y_3-y_2)(y_4-y_3)}}{4\sqrt{\pi} y^2_3 \rho_2(\rho_2y_3-\phi_2)Q_X(y_3)P_X\left(\frac{\phi_2}{\rho_2 y_3}\right)} \frac{1}{n \sqrt{n} y_3^n} \;;
\end{equation}
\item[(d)] If $\phi_1\leq \rho_1$ (and then $\phi_2>\rho_2$) ,
\begin{equation}
\label{equivd}
P(N_2 = n) \sim  \frac{(\rho_1-\phi_1)(\phi_2-\rho_2)\left((1-\phi_1)\mu_1+(1-\phi_2)\mu_2  \right)}{(1-\phi_1)\left(\phi_2(\rho_1-\phi_1)\mu_1+(1-\phi_2)(\phi_2-\rho_2)\mu_2\right)} \left( \frac{\rho_2}{\phi_2} \right)^n,
\end{equation}
\end{itemize}
where the constant $\kappa$ is given by
$$
\kappa =  \phi_2^2\mu_2^2(1-\phi_2)   \left((1-\phi_1)\mu_1+(1-\phi_2)\mu_2\right) P(0,0) .
$$
\end{proposition}

Before proceeding to the proof of Proposition~\ref{asymptoticest}, it is worth noting that the tail of the distribution of the number of jobs  in queue \#2 intricately depends on all the parameters of the system. This phenomenon has already been observed for the work-conserving GPS queue in \cite{GP}.

In addition, when $\phi_1+\phi_2=1$, it is readily checked that we can recover from the above formulas the results established in \cite{GP}; in that case $\varphi_Y(y) \equiv 1$. Thus, the asymptotic results stated in Proposition~\ref{asymptoticest} are valid for $\phi_1+\phi_2\geq 1$.

Finally, note that in case $(b)$, the value of $\varphi_Y(y_3)$ involves a Cauchy integral \cite{Lions} since the point $y_3$ is on the integration contour defining $\varphi_Y(y)$.

\begin{proof}[Proof of Proposition~\ref{asymptoticest}]
In case (d), the radius of convergence of $P(0,y)$ is equal to $y_3$ and $\rho_2/\phi_2$  is the root with the smallest module of the function $P(1,y)$. A  direct application of Darboux method \cite{Henrici} yields 
$$
P(N_2 = n)  \sim \frac{1}{\phi_2}         \left(\phi_2P(1,0)-(1-\phi_2)P\left(0,\frac{\phi_2}{\rho_2}  \right)+(1-\phi_2)P(0,0)\right) \left(\frac{\rho_2}{\phi_2}\right)^n.
$$
Using Equation~\eqref{residuP0y}, we obtain Equation~\eqref{equivd}.

In case (a), we know from Lemma~\ref{removepole} that $\phi_2/\rho_2$ is a removable singularity for the generating function $P(1,y)$. The point $\xi_+>\sqrt{\frac{\phi_2}{\rho_2}}$ is a pole for the function $P(0,y)$ which reads for $y>\sqrt{\frac{\phi_2}{\rho_2}}$
$$
P(0,y) =\frac{-\phi_2P(0,0)}{1-\phi_1-\phi_2} \frac{h_3\left(X^*(y),\frac{\phi_2}{\rho_2 y}  \right)h_2\left(X^*(y),y \right)}{h_2\left(X^*(y),\frac{\phi_2}{\rho_2 y}  \right)h_3\left(X^*(y),y \right)}\varphi_Y(y)+\frac{(1-\phi_1)P(0,0)}{1-\phi_1-\phi_2}.
$$
The residue of the function $P(0,y)$ at point $\xi_+$ is 
$$
 \frac{-\phi_2P(0,0) r(\xi_+) \varphi_Y(\xi_+)}{(1-\phi_1-\phi_2)},
$$
where
$$
r(\xi_+) = \frac{ h_3\left(X^*(\xi_+),\frac{\phi_2}{\rho_2 \xi_+}  \right)h_2\left(X^*(\xi_+),\xi_+ \right)}{h_2\left(X^*(\xi_+),\frac{\phi_2}{\rho_2 \xi_+}  \right)\left( \frac{\partial  h_3}{\partial x}\left(X^*(\xi_+),\xi_+ \right)\frac{d X^*}{dx}(\xi_+) + \frac{\partial  h_3}{\partial y}\left(X^*(\xi_+),\xi_+ \right)      \right)}. 
$$
A  direct application of Darboux method yields Equation~\eqref{equiva} since in the neighborhood of $\xi_+$
$$
P(1,y) \sim \frac{\phi_2(1-\phi_2)P(0,0)r(\xi_+)\varphi_Y(\xi_+)}{ (1-\phi_1-\phi_2)(\phi_2-\rho_2\xi_+) (y-\xi_+)}. 
$$

In cases (b) and (c), the function $\alpha_Y(y)$ has no poles in the disk with center 0 and radius $y_3$ and we can write for some $\delta \in \left(\frac{\phi_2}{\rho_2},y_3\right)$ and $y$ in the disk with center 0 and radius $\delta$
$$
P(0,y) = \frac{1}{2i\pi} \int_{C(0,\delta)}\frac{-\phi_2P(0,0)}{1-\phi_1-\phi_2}\alpha_Y(z)\varphi_Y(z)\frac{dz}{z-y} +\frac{(1-\phi_1)P(0,0)}{1-\phi_1-\phi_2},
$$
where $C(0,\delta)$ is the circle with center 0 and radius $\delta$.  By using the fact that the point $\phi_2/\rho_2$ is a removable singularity for the function $P(1,y)$, we have
$$
\phi_2P(1,0)-(1-\phi_2)P\left(0,\frac{\phi_2}{\rho_2}\right)+(1-\phi_2)P(0,0)=0
$$
and then
$$
P(1,y) = \frac{1-\phi_2}{\phi_2} \frac{P\left(0,\frac{\phi_2}{\rho_2}\right)-P(0,y)}{1-\frac{\rho_2 y}{\phi_2}} =\frac{1}{2i\pi} \int_{C(0,\delta)}h_Y(z)\frac{dz}{z-y},
$$
where
$$
h_Y(z) =\frac{-\phi_2(1-\phi_2)P(0,0)}{(1-\phi_1-\phi_2)(\rho_2z-\phi_2)}\alpha_Y(z)\varphi_Y(z).
$$
Noting that the point $\xi_-$ can be a pole for the function $h_Y(y)$ (with residue $r(\xi_-)$) and that the function $h_Y(y) =O(1/|y|)$ when $|y|\to \infty$, we obtain by deforming the integration contour $C(0,\delta)$
\begin{multline*}
P(1,y) =\\ \frac{1}{\pi} \int_{y_3}^{y_4} \frac{-\phi_2(1-\phi_2)P(0,0)}{(1-\phi_1-\phi_2)(\rho_2z-\phi_2)}\varphi_Y(z) \Im(\alpha_Y(z+0i)))\frac{dz}{z-y} -\frac{r(\xi_-)}{y-\xi_-}.
\end{multline*}

We have
$$
 \Im(\alpha_Y(y+0i)) = \frac{\Im\left(H_2(y)H_3(y) \right)   }{\left|h_3\left(X^*(y+0i),y \right)\right|^2\left| h_2\left({X^*(y+0i)},\frac{\phi_2}{\rho_2 y}\right)\right|^2      }
$$
where
\begin{eqnarray*}
H_3(y) &= & h_3\left(X^*(y+0i),\frac{\phi_2}{\rho_2 y}\right)  h_3\left(\overline{X^*(y+0i)},y\right),\\ 
H_2(y) &=& h_2\left(\overline{X^*(y+0i)},\frac{\phi_2}{\rho_2 y}\right)  h_2\left({X^*(y+0i)},y\right) . 
\end{eqnarray*}

For $y \in [y_3,y_4]$, we have
\begin{multline*}
\left|h_3\left(X^*(y+0i),y\right)\right|^2 =  h_3\left(X^*(y+0i),y)\right)h_3\left(\frac{\phi_1}{\rho_1 X^*(y+0i)}  \right)\\
=\left(((\phi_1\mu_1-(1-\phi_2)\mu_2)y+(1-\phi_2))X^*(y+0i) -\phi_1\mu_1y \right)\times \\  \left(((\phi_1\mu_1-(1-\phi_2)\mu_2)y+(1-\phi_2))\frac{\phi_1}{\rho_1 X^*(y+0i)} -\phi_1\mu_1y\right)
\end{multline*}
By using the fact that
$$
X^*(y+0i)+ \frac{\phi_1}{\rho_1 X^*(y+0i)} = -\frac{1}{\lambda_1 y}(\lambda_2y^2-(\lambda_1+\lambda_2+\phi_1\mu_1+\phi_2\mu_2)y + \phi_2\mu_2),
$$
we deduce that the function $y \to \left|h_3\left(X^*(y+0i),y\right)\right|^2$ is a cubic polynomial in variable $y$. The coefficient of the leading term is 
$$
\frac{\lambda_2\phi_1\mu_1}{\lambda_1}(\phi_1\mu_1-(1-\phi_2)\mu_2).
$$
The point 1 is obviously a root of this cubic polynomial. The other roots are such that $h_3(X^*(y),y)=0$ and are then the roots $\xi_\pm$ of the resultant $\mathrm{Res}_x(h_1,h_3)(y)$. It follows that 
$$
\left|h_3\left(X^*(y+0i),y\right)\right|^2 = \frac{\phi_1}{\rho_1}(y-1)Q_X(y).
$$

By using the same kind of arguments and the fact that $X^*(y) = X^*(\phi_2/(\rho_2 y))$, we have
$$
\left| h_2\left({X^*(y+0i)},\frac{\phi_2}{\rho_2 y}\right)\right|^2 = -\frac{1}{\rho_1}\left( \frac{\phi_2}{\rho_2 y}-1 \right)P_X\left(\frac{\phi_2}{\rho_2 y}\right).
$$

Furthermore, tedious calculations show that  
\begin{multline*}
\Im\left(H_2(y)H_3(y) \right) =   -\frac{\phi_1\phi_2\mu_1^2\mu_2^2}{\lambda_1}(1-\phi_1-\phi_2)\left((1-\phi_1)\mu_1+(1-\phi_2)\mu_2\right)  \times \\ \left(y-\frac{\phi_2}{\rho_2 y}  \right)(y-1)\left(y-\frac{\phi_2}{\rho_2}  \right)\frac{\sqrt{-\Delta_2(y)}}{2\lambda_1 y^2}, 
\end{multline*}
where we have used the fact that
$$
\Im(X^*(y+0i)) = \frac{\sqrt{-\Delta_2(y)}}{2\lambda_1 y}
$$
for $y \in [y_3,y_4]$. Hence, for $y \in [y_3,y_4]$
\begin{multline*}
\Im(\alpha_Y(y+0i) ) = \\ -\frac{\phi_2\mu_2^2(1-\phi_1-\phi_2)\left((1-\phi_1)\mu_1+(1-\phi_2)\mu_2\right)   \left(y-\frac{\phi_2}{\rho_2 y}  \right)}{2yQ_X(y)P_X\left(\frac{\phi_2}{\rho_2 y}\right)}\sqrt{-\Delta_2(y)}
\end{multline*}
and we deduce that
$$
P(1,y) =  \frac{1}{\pi} \int_{y_3}^{y_4} H_Y(z) \frac{dz}{z-y} +\frac{r(\xi_-)}{y-\xi_-},
$$
where the function $H_Y(y)$ is defined by
$$
H_Y(y) =   \kappa \frac{ \left(y- \frac{\phi_2}{\rho_2 y} \right)}{2y(\rho_2y-\phi_2)Q_X(y)P_X\left(\frac{\phi_2}{\rho_2 y}\right)}\varphi_Y(y)\sqrt{-\Delta_2(y)}
$$
with 
$$
\kappa =  \phi_2^2\mu_2^2(1-\phi_2)   \left((1-\phi_1)\mu_1+(1-\phi_2)\mu_2\right) P(0,0) .
$$
From the above computations, we deduce that
$$
P(N_2 = n) =  \frac{1}{\pi} \int_{y_3}^{y_4} \frac{H_Y(z)}{z} e^{-n\log z} {dz} +\frac{r(\xi_-)}{\xi_-^{n+1}} .
$$

By using the same arguments as in \cite{GP}, we easily deduce estimates~\eqref{equivc} and \eqref{equivb}.
\end{proof}

\section{Conclusion}
\label{conclusion}
After having established the stability conditions for a non work-conserving GPS queuing system, we have formulated the Riemann-Hilbert problems appearing when computing the generating function of the numbers of jobs in the system in the stationary regime. It turns out that these Riemann-Hilbert problems have indexes equal to 0 and can  explicitly be solved. This allows us to eventually compute  the above generating function. Using the analytic formulas, it is then possible to derive the queue asymptotics. It is amazing to observe that the queue asymptotics obtained for the non work-conserving GPS system are similar to those of  a work conserving GPS system.

\appendix

\section{Poles of the function $P(0,y)$}

To determine the poles of the function $P(0,y)$ we are led to study the roots of the equations $h_3(X^*(y),y)=0$ and $h_2\left(X^*(y),\frac{\phi_2}{\rho_2 y}\right)=0$. We  precisely have the following results.

\begin{lemma}
\label{lemmah1h2}
The equation $h_3(X^*(y),y)=0$ has a unique solution  in the interval $\left(\sqrt{\frac{\phi_2}{\rho_2}},y_3\right)$ if and only if $\phi_1>\rho_1$ and $Q_Y\left(\sqrt{\frac{\phi_1}{\rho_1}}  \right)<0$. In that case, the solution is $\xi_+$ defined by  Equation~\eqref{xipm}.
\end{lemma}

\begin{proof}
The couple $(x,y)$ is solution to the equations $h_1(x,y)=0$ and $h_3(x,y)=0$ if the hyperbolic curve 
$$
x= \frac{\phi_1 \mu_1y}{(\phi_1\mu_1-(1-\phi_2)\mu_2)y+(1-\phi_2)\mu_2}
$$
intersects the branches $x = X_\pm(y)$ at point $y$. For $y \in \left(y_2,y_3\right)$, the curves $x=X_\pm(y)$ clearly delineate a closed domain $D_X$. The above hyperbolic curve intersects the curves $x= X_\pm(y)$ at point $y=1$ and at another point denoted by $y^*$. The intersection point is located on the branch $X^*(y) = X_-(y)$ if $\zeta^* = X^*(\xi^*)< \sqrt{\frac{\phi_1}{\rho_1}}$. This is equivalent to the fact that the resultant $ \mathrm{Res}_y(h_1,h_2)(y)$ has a root between $1$ and $\sqrt{\frac{\phi_1}{\rho_1}}$, that is, $\phi_1>\rho_1$ and $Q_Y\left(\sqrt{\frac{\phi_1}{\rho_1}}  \right)<0$.  When this condition is satisfied, $\xi^*>1$. Otherwise, we would have  $\zeta^*\leq 1$ and the couple $(\zeta^*,\xi^*)$ would be a solution to  $h_1(x,y)=h_3(x,y)=0$ in $\overline{D(0,1)}\times\overline{D(0,1)}$, which is no possible when the system is stable.  In addition, it can be shown that $\xi_+\geq \sqrt{\frac{\phi_2}{\rho_2}}$.
\end{proof}

 Now, for the zeros of the function $h_2\left(X^*(y),\frac{\phi_2}{\rho_2 y}\right)$, we have the following result.

\begin{lemma}
\label{lemmah1h3}
The equation $h_2\left(X^*(y),\frac{\phi_2}{\rho_2 y}\right) =0$ has no solutions  in the interval  $\left(\sqrt{\frac{\phi_2}{\rho_2}},y_3\right)$.
\end{lemma}

\begin{proof}
Let us first note that for $y \in C\left(0,\sqrt{\frac{\phi_2}{\rho_2}}\right)$ we have $X^*(y) = X^*(\overline{y}) = X^*\left(\frac{\phi_2}{\rho_2 y}\right)$. Since the function $X^*(y)$ is analytic in $\C\setminus([y_1,y_2]\cup[y_3,y_4])$, the identity $X^*(y) =  X^*\left(\frac{\phi_2}{\rho_2 y}\right)$ holds for all $y \in \C\setminus([y_1,y_2]\cup[y_3,y_4])$. Hence, we have  $h_2\left(X^*(y),\frac{\phi_2}{\rho_2 y}\right) =0$ if $\xi = \frac{\phi_2}{\rho_2 y}$ is solution to the equation $h_2(X^*(\xi),\xi)=0$. The point $y$ is in  $\left(\sqrt{\frac{\phi_2}{\rho_2}},y_3\right)$ if $\xi$ is in $\left(1,\sqrt{\frac{\phi_2}{\rho_2}}\right)$, but this is not possible if the system is stable.
\end{proof}

\end{document}